\documentclass[10pt,twocolumn,a4paper]{IEEEtran}
\usepackage{graphicx}
\usepackage{caption,color}
\usepackage{refstyle}
\usepackage{amsmath}
\usepackage{amsfonts}
\usepackage{subfigure}
\usepackage{multirow}
\usepackage{algorithm}
\usepackage{amssymb}
\usepackage{cite}
\usepackage{algorithmic}
\usepackage{epstopdf}
\usepackage{cite}
\usepackage{booktabs}
\usepackage{slashbox}
\usepackage{amsthm}
\usepackage{subfigure}

\newcommand\numberthis{\addtocounter{equation}{1}\tag{\theequation}}

\newtheorem{prop}{Proposition}

\DeclareMathOperator*{\argmax}{arg\,max}
\ifCLASSINFOpdf
\else
\fi
\begin{document}

\linespread{1.0}
\title{Low Complexity Coefficient Selection Algorithms for Compute-and-Forward  }
\author{Qinhui Huang and Alister Burr
\thanks{The work described in this paper was supported in part by UK EPSRC under grant EP/K040006.\par  The authors are with the Department of Electronics, University of York, York
YO10 5DD, U.K. (e-mail: qh529@york.ac.uk; alister.burr@york.ac.uk) }}

\maketitle

\begin{abstract}
Compute-and-Forward (C\&F) has been proposed as an efficient strategy to reduce the backhaul load for the distributed antenna systems. Finding the optimal coefficients in C\&F has commonly been treated as a shortest vector problem (SVP), which is N-P hard. The point of our work and of Sahraei's recent work is that the C\&F coefficient problem can be much simpler. Due to the special structure of C\&F, some low polynomial complexity optimal algorithms have recently been developed. However these methods can be applied to real valued channels and integer based lattices only. In this paper, we consider the complex valued channel with complex integer based lattices. For the first time, we propose a low polynomial complexity algorithm to find the optimal solution for the complex scenario. Then we propose a simple linear search algorithm which is conceptually suboptimal, however numerical results show that the performance degradation is negligible compared to the optimal method. Both algorithms are suitable for lattices over any algebraic integers, and significantly outperform the lattice reduction algorithm. The complexity of both algorithms are investigated both theoretically and numerically. The results show that our proposed algorithms achieve better performance-complexity trade-offs compared to the existing algorithms.

\end{abstract}


\section{Introduction}
Due to their very high density, the next generation of wireless communication systems will require enormous backhaul load to support the data transmission between the access points and the central hub station. Physical layer network coding (PNC) \cite{Zhang.2006} has been proposed as a promising strategy to reduce the backhaul load. Among many PNC schemes, compute and forward (C\&F), as proposed in \cite{Nazer.2011} has attracted the most interest. It employs a structured lattice code for PNC. Each relay infers and forwards a linear combination of the transmitted codewords of all users. The lattice structure ensures the combination of the codewords is a codeword itself; hence cardinality expansion is avoided. Additionally, the abundant members of the ``lattice family" brings more flexibility to PNC.\par
The key aspect which dominates the performance of C\&F is the selection of the coefficient vectors. The process of selecting the optimal coefficients consists of two stages:
\begin{itemize}
\item local selection: each relay selects an integer vector to maximise its computation rate (achievable rate region) locally.
\item global selection: in order to recover the data without ambiguity, the vectors provided by the relays have to form a matrix whose rank is at least the number of sources.   
\end{itemize}
\par
Much work has been carried out in the last few years on both stages. For the local selection, the original paper of C\&F \cite{Nazer.2011} provided a bound for the coefficient vectors, and the optimal solution can be obtained by performing an exhaustive search within that boundary. The authors in \cite{Feng.2013} stated that the coefficient selection issue is actually a shortest vector problem (SVP). Any lattice reduction algorithm, such as the Lenstra-Lenstra-Lovasz (LLL) algorithm \cite{LLL.1982} and the Fincke-Pohst algorithm \cite{Fincke.1985} can be utilised to acquire the sub-optimal solution. There are two main drawbacks of these lattice reduction algorithms: 1) the complexity increases exponentially as the number of user terminals increases. 2) it becomes less accurate for large numbers of users. In 2014, a polynomially optimal algorithm proposed by Sahraei and Gastpar \cite{Gastpar.2014} significantly reduced the number of candidate vectors of \cite{Nazer.2011}. It translated the optimisation problem over multiple variables to one variable. Based on the idea of \cite{Gastpar.2014}, some improvements are proposed in \cite{Wen.2016, Qinhui.2016} to further reduce the complexity.   
\par 
Unfortunately, the methods in \cite{Gastpar.2014, Wen.2016, Qinhui.2016} are suitable for real valued channels and integer lattices ($\mathbb{Z}$-lattice) only. Finding the optimal solution in polynomial time over complex integer based lattices is still an open problem. For the Gaussian integer\footnote{Gaussian integers are complex numbers whose real and imaginary parts are both integers.} ($\mathbb{Z}[i]$) based lattices, the sub-optimal lattice reduction based algorithms: such as the complex-LLL  \cite{CongLing.2009} and its extensions \cite{Stewart.2001}, \cite{Cohen.2013} still work. However, they have the same drawbacks as in the real channel scenarios. Recently, much focus was given to the Eisenstein integer\footnote{Eisentein integers are complex numbers of the form
$c=a+b\omega$ where a and b are integers and $\omega=\frac{1}{2}(-1+\sqrt{3}i)$} ($\mathbb{Z}[\omega]$) based lattice: which has the densest packing strcuture in the 2-dimensional complex plane \cite{Tunali.2015, Wang.2015, Qifu.2013}. A lattice reduction method over the $\mathbb{Z}[\omega]$-lattice is proposed in \cite{Qifu.2013}, though for a two way relay system only. An extended version of the algorithm in \cite{Gastpar.2014} for both $\mathbb{Z}[i]$ and $\mathbb{Z}[\omega]$ is proposed in \cite{Ling.2016}, however it might miss the optimal solution sometimes. The latest research in \cite{Huang.2015} illustrated that the C\&F can be operated over many algebraic number fields (not only restricted to Gaussian and Eisenstein integers). Unfortunately, efficient approaches for coefficient selection over these non-cubic lattices are not available in the existing literature.\par 
For the second stage, the most commonly used approach to meet the requirement of unambiguous decodability is: each relay forwards more than one linear equation to the hub. The global optimal full rank matrix is selected by the hub and then fedback to the relays \cite{Liliwei.2012}. An alternative approach is that the integer vector provided by each relay is forced to include at least two users. This can significantly reduce the possibility of rank deficiency \cite{Molu.2016}. \par     
Distributed massive MIMO (or cell free massive MIMO)\cite{Cellfree.C, Cellfree.J} is probably the most promising application of C\&F. It deploys many more access points than user terminals. By exploiting the ``redundant" relays, the rank deficiency is not a big issue even if each relay forwards only the locally best equation without feedback\cite{Qinhui.2017}. Therefore in this paper, we focus on choosing the local optimal coefficient since it plays a fundamental role in the entire process of C\&F. The main contributions of this paper are as follows:
\begin{itemize}
\item For the first time, we propose a low polynomial complexity algorithm to ensure the optimal integer vector can be acquired for both $\mathbb{Z}[i]$ and $\mathbb{Z}[\omega]$ lattices. We also derive a theoretical upper bound of the complexity.  
\item We propose a suboptimal linear search algorithm for the coefficient selection which has lower complexity. Compared to the optimal approach above, it aims to discard the ``unnecessary" candidates by employing a pre-defined step size which is related to the number of users and SNR. The theoretical complexity is also investigated.   
\item We evaluate the performance and complexity of our proposed two algorithms numerically, and compare them with other existing approaches. Simulation results indicate that our proposed algorithms have better complexity-performance tradeoff. 
\item Our proposed algorithms can be easily extended to the lattices over any other algebraic integers without additional complexity.
\end{itemize} 
\par                      
The rest of this paper is organised as follows. We review the C\&F strategy and some existing selection algorithms as benchmarks in section II. In section III, we propose an optimal search and analyse its complexity. We introduce our linear search method and analyse its complexity in Section IV. In section V, we give the numerical results in terms of both computation rate and complexity for different types of lattices. Conclusions and future work are given in section VI. 

Unless noted, we use plain letters, boldface lowercase letters and boldface uppercase letters to denote scalars, vectors and matrices respectively, and all vectors are column vectors. The sets of real numbers and complex numbers are denoted by $\mathbb{R}$ and $\mathbb{C}$ respectively. We use $\mathbb{Z}$, $\mathbb{Z}[i]$ and $\mathbb{Z}[\omega]$ to represent integers, Gaussian integers and Eisenstein integers respectively. $\mathbb{F}_p$ denotes the finite field of size $p$. $\lfloor\cdot\rceil$, $\lceil\cdot\rceil$, $\lfloor\cdot\rfloor$ denote the round, ceil and floor operations respectively. We use $\|\cdot\|$ to represent the Euclidean norm.  

\section{Preliminaries}
\subsection{Compute and Forward}
We consider a general local optimisation problem in C\&F. As shown in Fig. 1, we assume that $L$ users transmit signals to the relay simultaneously. The original transmitted message of the $l$-th user is denoted as $\mathbf{w}_l\in\mathbb{F}_p^k$, which is a length $k$ vector over GF($p$). By employing a $\frac{k}{n}$-rate lattice encoder, $\mathbf{w}_l$ is mapped to a length $n$ codeword, denoted $\mathbf{x}_{l}\in{\mathbb{C}}^{n}$. Each component of $\mathbf{x}_l$ is drawn from a quotient ring $\mathbb{A}/{\pi}\mathbb{A}$ which is isomorphic to GF($p$). The term $\mathbb{A}$ denotes an integer domain, and usually refers to a principal ideal domain (PID)\footnote{The most commonly used PIDs for complex valued case are Gaussian integers $\mathbb{Z}[i]$ and Eisenstein integers $\mathbb{Z}[\omega]$, hence their respective $\mathbf{x}_{l}$ can be expressed as $\mathbf{x}_{l}\in{(\mathbb{Z}[i]/{\pi}\mathbb{Z}[i])}^{n}$ and $\mathbf{x}_{l}\in{(\mathbb{Z}[\omega]/{\pi}\mathbb{Z}[\omega])}^{n}$.}. The codebook of $\mathbf{x}_{l}$ is defined by a lattice partition of $\Lambda/{\Lambda'}$\footnote{$\Lambda$ and $\Lambda'$ denote the fine lattice and the coarse lattice respectively. Note that $\mathbb{A}/\pi{\mathbb{A}}$ corresponds to symbol, whereas $\Lambda/{\Lambda'}$ corresponds to codeword, their respective cardinalities are $p$ and $p^k$. The design of the codebook is beyond the scope of this paper, see \cite{Feng.2013} for details.}. We use Rayleigh fading $h_{l}\sim{\mathcal{CN}(0,1)}$ to model the channel vector $\mathbf{h}=[h_1,h_2,\cdots,h_L]^T$. The received superimposed signal at the relay can be expressed as
\begin{equation}
\mathbf{y}=\sum_{l=1}^L{h_l\mathbf{x}_l}+\mathbf{z},\ \mathbf{y}\in\mathbb{C}^n,
\end{equation}
 where the noise $\mathbf{z}\sim\mathcal{CN}(0,\sigma^2{\mathbf{I}_{n}})$ is a length $n$ circularly symmetrical complex Gaussian random vector. We assume the power constraint of the codeword is $P$ per symbol, written as $E[\|\mathbf{x}_{l}\|^2]\leq{nP}$. The signal to noise ratio is represented as SNR=$P/{\sigma^2}$. 
\begin{figure}[t!]
\centering
\begin{minipage}[t]{1\linewidth}
\centering
\includegraphics[width=0.9\textwidth]{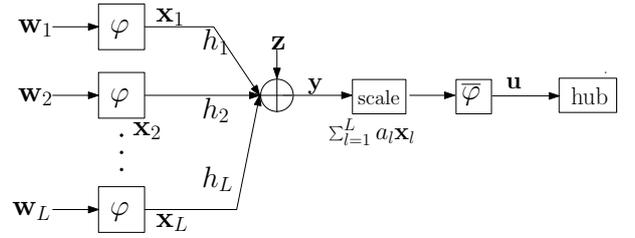}
\end{minipage}
\caption{The system model of compute and forward } \label{fig:buffer1}
\end{figure}

The received signal vector $\mathbf{y}$ is first scaled by a factor $\alpha\in\mathbb{C}$. Each relay attempts to choose an integer linear combination of the transmitted codewords to represent the scaled received signal, written as $\mathcal{Q}_{\Lambda}(\alpha\mathbf{y})=\sum_{l=1}^{L}{a_{l}\mathbf{x}_{l}}$. $\mathcal{Q}_{\Lambda}$ quantises $\alpha\mathbf{y}$ to its closest fine lattice point in $\Lambda$. The quantisation error contributes to the effective noise of C\&F, expressed as 
\begin{equation}
\mathbf{z_{\mathrm{eff}}}=\displaystyle{\sum_{l=1}^L}(\alpha h_l-a_l)\mathbf{x}_l+\alpha\mathbf{z},
\end{equation}
where $a_{l}$ is an integer in $\mathbb{A}$. Let $\mathbf{a}=[a_1,a_2,\cdots,a_L]^T$ to denotes the coefficient vector of the linear function. The scaling factor $\alpha$ aims to force the scaled channel vector $\alpha \mathbf{h}$ to approximate an integer vector $\mathbf{a}$. The effective noise comprises 2 components:
\begin{itemize}
\item Self noise: $\mathbf{z_{\mathrm{self}}}=\sum_{l=1}^{L}(\alpha h_l-a_l)\mathbf{x}_l$: caused by the mismatch between the selected integer vector and the scaled channel.
\item Scaled Gaussian noise $\mathbf{z_{\mathrm{sg}}}=\alpha \mathbf{z}$: the received Gaussian noise is scaled by the scaling factor $\alpha$.
\end{itemize}
For a given coefficient vector $\mathbf{a}$, the achievable computation rate per complex dimension is given as \cite{Nazer.2011}\\
\begin{equation}
\mathcal{R}(\mathbf{h},\mathbf{a})=\displaystyle{\max_{\alpha\in\mathbb{C}}}\log^{+}\Big(\frac{P}{\alpha^2\sigma^2+P\|\alpha\mathbf{h}-\mathbf{a}\|^2}\Big),
\end{equation}
where $\mathrm{log}^{+}(\cdot)=\mathrm{max}(\mathrm{log}(\cdot),0)$, and the term $\alpha^2\sigma^2+P\|\alpha\mathbf{h}-\mathbf{a}\|^2$ is the variance of the effective noise, denoted by $\sigma_{\mathrm{eff}}^2$. The Minimum Mean Square Error (MMSE) solution of $\alpha$ is given by\\
\begin{equation}
\alpha_\mathrm{MMSE}=\frac{\mathrm{SNR}\mathbf{h}^H\mathbf{a}}{1+\mathrm{SNR}\parallel\mathbf{h}\parallel^2}, 
\end{equation}
and hence equation (3) can be rewritten as
\begin{equation}
\mathcal{R}(\mathbf{h},\mathbf{a})=\log^{+}(\frac{1}{\mathbf{a}^{H}\mathbf{M}\mathbf{a}}),
\end{equation}
where $\mathbf{M}=\mathbf{I}_{L}-\frac{\mathrm{SNR}}{\mathrm{SNR}\|\mathbf{h}\|^{2}+1}\mathbf{h}\mathbf{h}^{H}$, and $\mathbf{I}_L$ denotes an $L\times{L}$ identity matrix. 
The target of each relay is to find its local best integer vector $\mathbf{a}$ to maximise the computation rate, expressed as
\begin{equation}
\mathbf{a_{\mathrm{opt}}}=\argmax_{\mathbf{a}\in{\mathbb{A}}^{L}\setminus\{\mathbf{0}\}}\mathcal{R}(\mathbf{h},\mathbf{a}).
\end{equation} 
\subsection{Existing Coefficient Selection Algorithms}
\subsubsection{Exhaustive-I Algorithm}
In the original paper of C\&F\cite{Nazer.2011}, the authors stated that the Euclidean norm of the optimal coefficient vector has an upper bound, written as $\|\mathbf{a}_{\mathrm{opt}}\|\leq\Phi=\sqrt{1+\mathrm{SNR}\|\mathbf{h}\|^2}$, hence an exhaustive search over all possible $\mathbf{a}$ within that range can be employed to obtain $\mathbf{a}_{\mathrm{opt}}$. The time complexity of this algorithm is $\mathcal{O}(\Phi^{2L})$.
\subsubsection{Exhaustive-II Algorithm (Real-valued only)} 
The authors in \cite{Gastpar.2014, Gastpar.2015} proposed an exhaustive search algorithm with polynomial complexity\footnote{In this paper, we focus on the the complex valued case only. Hence some improved versions of this method are omitted here, see \cite{Wen.2016, Qinhui.2016} for details.}. They stated that it suffices to search over the integer vectors generated by $\lfloor\alpha\mathbf{h}\rceil$ only rather than considering all possible $\mathbf{a}$ in $\mathbb{Z}^{L}$. Therefore, the optimisation problem with an $L$-dimensional variable $\mathbf{a}$ is translated to an optimisation problem over the one-dimensional variable $\alpha$. The candidate vectors can be obtained by dividing all possible $\alpha\in\mathbb{R}$ into several intervals, and each interval corresponds to a unique candidate $\mathbf{a}$. The time complexity of this algorithm is $\mathcal{O}(L\Phi\mathrm{log}(L\Phi))$. 
\subsubsection{Lattice Reduction Algorithm} Using lattice reduction based algorithms for coefficient selection was first proposed in \cite{Feng.2013}. As shown in equation (5), maximising the computation rate is equivalent to minimising $\mathbf{a}^{H}\mathbf{M}\mathbf{a}$. The matrix $\mathbf{M}$ can be decomposed as $\mathbf{M}=\mathbf{L}\mathbf{L}^H$ by employing the Cholesky decomposition. Hence the equation (5) can be rewritten as $\mathcal{R}(\mathbf{h},\mathbf{a})=\log^{+}(\frac{1}{\|\mathbf{L}^{H}\mathbf{a}\|^2})$. This is exactly a shortest vector problem (SVP) of an $L$-dimensional lattice generated by $\mathbf{L}^H$. The LLL and Complex-LLL lattice reduction algorithms are most commonly used for dealing with the SVP in $\mathbb{Z}$-lattice and $\mathbb{Z}[i]$-lattice respectively. However, these algorithms only ensure the selected vector is less than $2^{\frac{L-1}{2}}$ times the actual optimal solution. Hence, they become less accurate as the number of users increases.
\subsubsection{Quantised Search}: For $\mathbb{Z}[i]$-lattice, an intuitive approach for coefficient selection is to employ some quantised (sampled) values of $\alpha$ to generate the candidate set of $\mathbf{a}$, expressed as $\mathbf{a}=\mathcal{Q}_{\mathbb{Z}[i]}(\alpha\mathbf{h})$. The question is how to choose the quantiser. Since $\alpha\in\mathbb{C}$, the authors in \cite{Viterbo.2012}\footnote{Actually, the concept of utilising $\mathcal{Q}(\alpha\mathbf{h})$ instead of $\mathbf{a}$ was first proposed in \cite{Viterbo.2012}. However, rigorous prove and detailed analysis are not given in \cite{Viterbo.2012}} allocate step sizes for both the magnitude and the phase of $\alpha$. Clearly, this method is equivalent to the exhaustive search when both of the step sizes tend to zero. However, zero step size is definitely infeasible in practice. The core aspect of such a quantised algorithm is the choice of the step size, which is not analysed in \cite{Viterbo.2012}. \par
The method described above leads to an oversampling for the small magnitudes and undersampling for the large magnitudes. In section IV, we will propose an uniform quantiser and describe how to choose the optimal step size.
\subsubsection{L-L Algorithm}: Very recently, Liu and Ling proposed an efficient algorithm (denoted as L-L algorithm) for the complex valued channel in \cite{Ling.2016}. The authors adapted the idea in \cite{Gastpar.2014} directly for the complex integer based lattices. However, the algorithm in \cite{Ling.2016} does not ensure the selected coefficients are optimal for all channel realisations. A detailed discussion of this approach will be presented in section III-$C$.
\section{Exhaustive-II in Complex Valued Channel}
Since the Exhaustive-II selects the optimal coefficients with low polynomial complexity in the real channel case. Hence it is worthwhile to investigate the feasibility of Exhaustive-II in the complex valued channel. This section comprises three parts: we firstly propose the complex exhaustive-II algorithm, followed by the complexity analysis in section $B$, and then a comparison with the L-L method is given in section $C$.
 
\subsection{Complex Exhaustive-II Algorithm}
By substituting $\alpha\mathbf{h}$ for $\mathbf{a}$ in (3), the rate expression becomes
\begin{equation}
\mathcal{R}(\mathbf{h},\alpha)=\mathrm{log}^{+}\Big(\frac{P}{\alpha^2\sigma^2+P\|\alpha\mathbf{h}-\mathcal{Q}_{\mathbb{A}}(\alpha\mathbf{h})\|^2}\Big),
\end{equation} 
where $\alpha\in\mathbb{C}$. Actually, it is not necessary to evaluate $\alpha$ over the whole complex plane.
\begin{prop}
The amplitude of $\alpha_{\mathrm{opt}}$ is upper bounded by $\sqrt{\mathrm{SNR}}$, and it suffices to restrict the phases of $\alpha$ to $0\sim\frac{\pi}{2}$ and $0\sim\frac{\pi}{3}$ for $\mathbb{Z}[i]$-lattice and $\mathbb{Z}[w]$-lattice respectively.
\end{prop} 
\begin{proof}
According to (7), we have
\begin{align*}
\mathcal{R}(\mathbf{h},\alpha)&=\log^{+}\Big(\frac{P}{\alpha^2\sigma^2+P\|\alpha\mathbf{h}-\mathcal{Q}_{\mathbb{A}}(\alpha\mathbf{h})\|^2}\Big)\\
&\leq{\log^{+}(\frac{P}{\alpha^2\sigma^2})}=\log^{+}(\frac{\mathrm{SNR}}{\alpha^2}).\numberthis
\end{align*}
Apparently, the computation rate is zero when $\|\alpha\|\geq{\sqrt{\mathrm{SNR}}}$, where the equality holds iff the selected integer vector matches the scaled channel perfectly. Hence we have an upper bound of $\|\alpha_{\mathrm{opt}}\|<\sqrt{\mathrm{SNR}}$.\par 
Assume $u$ is a unit in $\mathbb{A}$, we have:
\begin{align*}
\mathcal{R}(\mathbf{h},\alpha)&=\log^{+}\Big(\frac{P}{\alpha^2\sigma^2+P\|\alpha\mathbf{h}-\mathcal{Q}_{\mathbb{A}}(\alpha\mathbf{h})\|^2}\Big)\\
&=\log^{+}\Big(\frac{P}{(u\alpha)^2\sigma^2+P\|u\alpha\mathbf{h}-\mathcal{Q}_{\mathbb{A}}(u\alpha\mathbf{h})\|^2}\Big)\\&=\mathcal{R}(\mathbf{h},u\alpha)\numberthis
\end{align*}
Hence the complex plane of $\alpha$ is divided into several ``equivalent regions" due to the existence of units. As the number of units in $\mathbb{Z}[i]$ and $\mathbb{Z}[w]$ are 4 and 6 respectively, it suffices to restrict the phase within $0\sim\frac{2\pi}{4}$ and $0\sim\frac{2\pi}{6}$ respectively.  
\end{proof}
\par
Recall the exhaustive-II search in the real channel case: the range of $\alpha\in\mathbb{R}$ is divided into several intervals. The quantised value $\lfloor\alpha\mathbf{h}\rceil$ is invariant within each interval. Hence each interval corresponds to an unique candidate vector $\mathbf{a}$ (the interval is called the Voronoi region or decision region of its corresponding $\mathbf{a}$), and the candidates can be acquired by choosing a representative of $\alpha$ for each interval. \par 
For the complex channel case, we use $h_{\mathrm{max}}$ to denote the channel coefficient with the largest amplitude in $\mathbf{h}$. Let $\upsilon_{0}$ denote the fundamental region of $\mathbb{A}$, and $\upsilon_{l,a_{l}}$ denote the Voronoi region of $\alpha$ for $\mathcal{Q}_{\mathbb{A}}(\alpha{h_{l}})=a_{l}$. Their respective areas are represented by $\mathcal{A}_{\upsilon_{0}}$ and $\mathcal{A}_{\upsilon_{l}}$. Note that the size of $\upsilon_{l,a_{l}}$ is invariant with different $a_{l}$. Clearly, we have
\begin{equation}
\mathcal{A}_{\upsilon_{l}}=\mathcal{A}_{\upsilon_{0}}/\|h_{l}\|^2,
\end{equation} 
hence we have the following results:
\begin{center}
\begin{figure}[h]
\centerline{\includegraphics[width=0.8\linewidth]{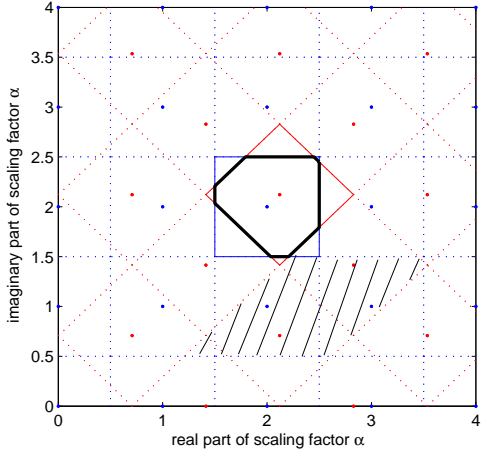}}
\caption{An example with $\mathbb{Z}[i]$: $\mathbf{h}=[1,\ \frac{1+i}{\sqrt{2}}]^T$, $\mathbf{a}=[2+2i,\ 3i]^T$}
\label{fig:buffer1}
\end{figure}
\end{center}

\begin{prop}
The complex plane of $\alpha$ is divided into several convex polygon regions, and each region corresponds to a unique vector $\mathbf{a}$. The area of each region is upper bounded by $\mathcal{A}_{\upsilon_{0}}/\|h_{\mathrm{max}}\|^2$. 
\end{prop} 
\begin{proof}
For a given candidate vector $\mathbf{a}=[a_1,a_2,\cdots,a_L]^{T}$, the value of $\alpha$ has to meet the conditions of $\mathcal{Q}_{\mathbb{A}}(\alpha\mathbf{h})=\mathbf{a}$ which is equivalent to 
\begin{equation}
\mathcal{Q}_{\mathbb{A}}(\alpha{h_{1}})=a_{1}\cap\mathcal{Q}_{\mathbb{A}}(\alpha{h_{2}})=a_{2}\cdots\cap\mathcal{Q}_{\mathbb{A}}(\alpha{h_{L}})=a_{L},
\end{equation} 
hence the Voronoi region for $\mathcal{Q}_{\mathbb{A}}(\alpha{\mathbf{h}})=\mathbf{a}$, denoted as $\upsilon_{\mathbf{a}}$ is the intersection region of $\upsilon_{l,a_{l}}$ for all $l$. Since the intersection of convex sets is also convex, hence the Voronoi region $\upsilon_{\mathbf{a}}$ is also a convex polygon (the shape of each individual $\upsilon_{l,a_{l}}$ depends on the fundamental region of $\mathbb{A}$: they are square and hexagon for $\mathbb{Z}[i]$ and $\mathbb{Z}[\omega]$ respectively). As an intersection of polygons, the area of $\upsilon_{\mathbf{a}}$ is therefore upper bounded by the smallest size among all $\upsilon_{l,a_l}$ which is $\mathcal{A}_{\upsilon_{0}}/\|h_{\mathrm{max}}\|^2$. 
\end{proof}

We take a simple example to interpret the above proposition. We consider a 2 user system employing the $\mathbb{Z}[i]$-lattice, with the channel vector  $\mathbf{h}=[1,\ \frac{1+i}{\sqrt{2}}]^T$. Since the fundamental region of $\mathbb{Z}[i]$ is square, therefore the shape of each $\upsilon_{l,a_{l}}$ is also square. As shown in Fig. 2, the real and imaginary parts of $\alpha$ are represented by the x-axis and y-axis respectively, and each red (blue) square corresponds to an unique $a_{1}$ ($a_2$) respectively. For example, the red (blue) solid square in the centre denotes $a_1=2+2i$ and $a_2=3i$ respectively. Hence in order to acquire $\mathcal{Q}_{\mathbb{Z}[i]}(\alpha\mathbf{h})=[2+2i, 3i]$, the value of $\alpha$ has to be chosen within the region of the black octagon in the centre.  
\par
\begin{algorithm}
\caption{Complex-Exhaustive-II Algorithm}
\begin{algorithmic}[1]
\REQUIRE channel vector $\mathbf{h}=[h_1,h_2,\cdots,h_L]\in\mathbb{C}^{L}$, SNR, integer domain $\mathbb{A}$ ($\mathbb{Z}[i]$,$\mathbb{Z}[\omega]$, etc) with basis $\mathbf{B}_{\mathbb{A}}$ 
\ENSURE optimal coefficient vector $\mathbf{a}_{\mathrm{opt}}$ \\*
$\mathbf{Phase\ 1}$: obtain the representatives of $\alpha$, stored in set $\mathcal{S}$. The initial $\mathcal{S}=\emptyset$
\STATE calculate the range of $\alpha$ according to Proposition. 1
\FOR {$l=1:L$}
\STATE find all lattice points generated by $\frac{1}{h_l}\mathbf{B}_{\mathbb{A}}$ over the range obtained in step.1. The acquired lattice points are stored in $\Omega_{l}=\{\alpha_{l,1}^{*},\alpha_{l,2}^{*},\cdots,\alpha_{l,K_{l}}^{*}\}$
\FOR {$k=1:K_{l}$}
\STATE find the vertices of the corresponding $\upsilon_{l,a_{l}}$ with $a_{l}=\alpha_{l,k}^{*}{h_{l}}$, calculated by $\alpha_{l,k}^{*}+\frac{1}{h_l}\frac{z_{0}}{2}$\\*
$z_{0}=\pm{1}\pm{i}$ for $\mathbb{Z}[i]$\\*
$z_{0}=\pm{1}\pm{\frac{\sqrt{3}}{3}}i,\ \pm\frac{2\sqrt{3}}{3}i$ for $\mathbb{Z}[\omega]$
\STATE store these vertices into set $\mathcal{S}_{l}$
\STATE calculate the linear equation of each edge of $\upsilon_{l}$, save them into set $\Psi_{l}$
\ENDFOR
\STATE $\mathcal{S}=\mathcal{S}\cup\mathcal{S}_{l}$
\ENDFOR
\FOR {$\bar{l}=1:L-1$}
\FOR {$\hat{l}=\bar{l}+1:L$}
\STATE find all combinations of $\{c_1,c_2\}$, with $c_1\in\Psi_{\bar{l}}$ and $c_2\in\Psi_{\hat{l}}$. Calculate the crossing point of $c_1$ and $c_2$: the crossing points which are not in the valid range of $\alpha$ should be discarded. Store the remaining in set $\mathcal{S}_{\bar{l},\hat{l}}$ 
\STATE $\mathcal{S}_{\bar{l}}=\mathcal{S}_{\bar{l}}\cup\mathcal{S}_{\bar{l},\hat{l}}$
\ENDFOR
\STATE $\mathcal{S}=\mathcal{S}\cup\mathcal{S}_{\bar{l}}$
\ENDFOR
\\*
$\mathbf{Phase \ 2}$ select the optimal integer vector 
\FOR {all representative $\alpha$ in $\mathcal{S}$}
\STATE acquire candidate of $\mathbf{a}$ by $\mathcal{Q}_{\mathbb{A}}^{*}(\alpha\mathbf{h})$, discard the repeated outputs.
\STATE calculate $\mathcal{R}(\mathbf{h},\mathbf{a})$ by equation (5)
\ENDFOR
\STATE Return $\mathbf{a}_{\mathrm{opt}}=\argmax\mathcal{R}(\mathbf{h},\mathbf{a})$
\end{algorithmic}
\end{algorithm}

The exhaustive-II requires the selection of a representative $\alpha$ within each polygon to obtain the candidate set of $\mathbf{a}$. In the real channel case, each $\mathbf{a}$ corresponds to an one-dimensional interval, therefore we can simply choose the end point (which is the discontinuity of the function $f(\alpha)=\lfloor\alpha\mathbf{h}\rceil$) of each interval as the representative. However, in the complex channel case, the one-dimensional interval becomes a two-dimensional region, the discontinuities become the edges of the polygon. Hence the number of discontinuities becomes infinite. Now the vertices of each polygon are most easily calculated among all discontinuities: can we therefore use these vertices as the representatives? \par 
Assume $\alpha_{\mathrm{v}}$ is a vertex, clearly, $\alpha_{\mathrm{v}}$ is shared by its adjacent polygons. Hence the element $\alpha_{\mathrm{v}}\mathbf{h}$ is singular to the quantisation operation $\mathcal{Q}_{\mathbb{A}}(\cdot)$ (due to the fact that at least one of $\mathrm{Real}(\alpha_{\mathrm{v}}\mathbf{h})$ and $\mathrm{Imag}(\alpha_{\mathrm{v}}\mathbf{h})$ is precisely a half integer). In the real valued channel with $\mathbb{Z}$-lattice, the singular quantisation is not a problem. Each interval has two ends, and hence if $\mathcal{Q}_{\mathbb{Z}}(\alpha\mathbf{h})$ is open at one end, then it has to be closed at the other end as long as $\mathcal{Q}_{\mathbb{Z}}$ rounds $\alpha\mathbf{h}$ in the same direction at both ends. This is because each interval has redundancy (2 ends) to compensate the quantisation uncertainty (2 possibilities: round up or down), and they are balanced for all intervals. However, for the complex channel, the redundancy and the quantisation uncertainty are not always balanced. Take the $\mathbb{Z}[i]$-lattice for example: each $\mathcal{Q}_{\mathbb{Z}[i]}(\alpha_{\mathrm{v}}\mathbf{h})$ has four possible values, while the number of vertices of each polygon is uncertain. Particularly, for the triangle regions, the redundancy (3 vertices) is apparently not able to compensate the quantisation uncertainty. This means if we set the quantiser to round $\alpha_{\mathrm{v}}\mathbf{h}$ in a specific direction for all vertices, we might miss that triangle polygon. Hence, we propose a ``full direction" quantiser $\mathcal{Q}_{\mathbb{A}}^{*}(\cdot)$ to replace $\mathcal{Q}_{\mathbb{A}}(\cdot)$. $\mathcal{Q}_{\mathbb{A}}^{*}(\cdot)$ returns all equal likely $\mathbf{a}$. For example, $\mathcal{Q}_{\mathbb{Z}[i]}^{*}(0.5+1.5i)=\{1+2i,1+1i,0+2i,0+1i\}$. The modified quantiser ensures there exists at least one representative within each polygon. In the next section we will see this modification only increases the complexity slightly.\par  
   
The only issue remaining is to calculate the coordinates of the vertices. Clearly, each vertex is a crossing point of two lines, and each line is exactly an edge of $\upsilon_{l,a_{l}}$. Since all $\upsilon_{l,a_{l}}$ have regular shapes, hence it is easy to acquire the function of each edge according to the coordinates of the centre point. The centre points are represented by the dots in Fig. 2, they are exactly the lattices points generated by the basis of $\frac{1}{h_{l}}\mathbf{B}_{\mathbb{A}}$, where $\mathbf{B}_{\mathbb{A}}$ is the basis matrix\footnote{The basis of $\mathbb{Z}[i]$ and $\mathbb{Z}[\omega]$ are $\mathbf{B}_{\mathbb{Z}[i]}=[1\ 0;0\ 1]$ and $\mathbf{B}_{\mathbb{Z}[w]}=[1\ 0;0\ \omega]$ respectively } of $\mathbb{A}$. Since the range of $\alpha$ is given at the beginning of this section, the centre points can be easily obtained. The whole procedure of the Complex-Exhaustive-II is summarised in Algorithm. 1.

\subsection{Complexity of Complex-Exhaustive-II Algorithm}

The complexity mainly depends on the number of candidates $\mathbf{a}$, and this number is upper bounded by the outputs of $\mathcal{Q}_{\mathbb{A}}^{*}(\alpha\mathbf{h})$ for all $\alpha$ in $\mathcal{S}$ (step 18-19 in Algorithm 1). Since the number of quantiser outputs for each $\alpha$, denoted as $\xi$ is a constant ($\xi=4$ for $\mathbb{Z}[i]$, $\xi=4$ or  6 for $\mathbb{Z}[\omega]$)\footnote{In principle, the possibility that more than 2 lines intersect at the same point is infinitesimal.}, hence the number of candidates $\mathbf{a}$ is bounded by $\xi{|\mathcal{S}|}$, where $|\mathcal{\cdot}|$ denotes the cardinality of a set. The $\alpha$ in $\mathcal{S}$ can be divided into 2 sets:
\begin{itemize}
\item $\mathcal{S}$-I: vertices of individual $\upsilon_{l,a_{l}}$ (step 2-10 in Algorithm 1).
\item $\mathcal{S}$-II: intersections of two sets of parallel lines, where one set belongs to $\Psi_{\bar{l}}$ and the other belongs to $\Psi_{\hat{l},\hat{l}\neq{\bar{l}}}$ (step 11-17 in Algorithm 1).  
\end{itemize} 
Intuitively, the former indicates the vertices of the red/blue squares in Fig.2, while the latter indicates the vertices of the parallelograms in Fig. 2 (labelled by the black shading). Since the area of the valid range of $\alpha$ is bounded by SNR, the total number of $\upsilon_{l,a_{l}}$ for all $l$ is therefore expected to be
\begin{equation}
\sum_{l}\mathbb{E}\Big[\frac{\mathrm{SNR}}{\mathcal{A}_{\upsilon_{l}}}\Big]=\sum_{l}\frac{\mathrm{SNR}\mathbb{E}[\|h_{l}\|^2]}{\mathcal{A}_{\upsilon_{0}}}=\frac{\mathrm{SNR}L}{\mathcal{A}_{\upsilon_{0}}},
\end{equation} 
where $\mathcal{A}_{\upsilon_{0}}$ is a constant as described previously, and the second equality is due to the assumption of $h_l\sim\mathcal{CN}(0,1)$. \par
For each pair of sets of parallel lines from $\Psi_{\bar{l}}$ and $\Psi_{\hat{l},\hat{l}\neq{\bar{l}}}$, the expected number of parallelograms is 
\begin{align*}
\mathbb{E}\Big[\frac{\mathrm{SNR}}{\mathcal{A}_{para}}\Big]&=\frac{\mathrm{SNR}\mathbb{E}[|h_{\bar{l}}||h_{\hat{l}}|\mathrm{sin}(\theta_{\bar{l},\hat{l}})]}{\mathcal{A}_{\upsilon_{0}}}\numberthis\\&=\frac{\mathrm{SNR}\mathbb{E}[|h_{\bar{l}}|]\mathbb{E}[|h_{\hat{l}}|]\mathbb{E}[\mathrm{sin}(\theta_{\bar{l},\hat{l}})]}{\mathcal{A}_{\upsilon_{0}}}\numberthis\\&=0.5\frac{\mathrm{SNR}}{\mathcal{A}_{\upsilon_{0}}}\numberthis.
\end{align*} 
Here $\mathcal{A}_{para}$ denotes the area of the parallelograms, and $\theta_{\bar{l},\hat{l}}$ denotes the intersection angle of the two sets of lines which is randomly distributed within $0\sim\frac{\pi}{2}$, hence $\mathbb{E}[\mathrm{sin}(\theta_{\bar{l},\hat{l}})]=\frac{2}{\pi}$. The expression (14) comes from the independence of the variables. Since the expected value of $|h|$ equals $\sqrt{\frac{\pi}{4}}$ with $h\sim\mathcal{CN}(0,1)$, the simplified expression is therefore written as (15). Since there are respectively 2 (3) sets of parallel lines for $\mathbb{Z}[i]$ ($\mathbb{Z}[\omega]$) in each $\Psi_{l}$, the total number of the parallelograms is therefore expected to be
\begin{equation}
\binom{L}{2}\binom{2}{1}\binom{2}{1}\mathbb{E}\Big[\frac{\mathrm{SNR}}{\mathcal{A}_{para}}\Big] \ \mathrm{and}\ \binom{L}{2}\binom{3}{1}\binom{3}{1}\mathbb{E}\Big[\frac{\mathrm{SNR}}{\mathcal{A}_{para}}\Big]
\end{equation} for $\mathbb{Z}[i]$ and $\mathbb{Z}[\omega]$ respectively. The expressions (12) and (16) also respectively represent the expected values of $|\mathcal{S}$-I$|$ and $|\mathcal{S}$-II$|$. Since the total number of candidates is $\xi(|\mathcal{S}$-I$|+|\mathcal{S}$-II$|)$ and the computation rate can be calculated in $\mathcal{O}(L)$ for each candidates, hence the overall time complexity can be expressed as
\begin{equation}
\mathcal{O}(\mathrm{SNR}L^{2}(L-1))+\mathcal{O}(\mathrm{SNR}L^2).
\end{equation} 
Note that the constant components are omitted in (17), and their effect will be evaluated numerically in section V. 

\subsection{L-L Algorithm vs Complex-Exhaustive-II Algorithm}
The L-L algorithm in \cite{Ling.2016} is described as an optimal deterministic algorithm.  Actually, it does not ensure the optimal solution for all channel realisations. In this section, we will present an example to compare the L-L algorithm and our proposed complex exhaustive-II algorithm. \par 

The main difference between these two algorithms is the elements of representative $\alpha$. The exhaustive-II algorithm considers both $\mathcal{S}$-I and $\mathcal{S}$-II, while the L-L algorithm considers the individual $\upsilon_{l,a_{l}}$ only. Specifically, the vertices and the midpoints of sides of individual $\upsilon_{l,a_{l}}$ are considered for L-L, hence the representatives of $\alpha$ can be regarded as an extended version of $\mathcal{S}$-I (though the L-L algorithm is not interpreted in such a manner in \cite{Ling.2016}). Fig.3 illustrates an intuitive comparison of these two algorithms. A $\mathbb{Z}[i]$-lattice based system is considered, with $L=5$ and SNR = 10dB. The channel components $h_{l}$ and their corresponding $\upsilon_{l,a_{l}}$ are denoted by different colours. The representatives of $\alpha$ utilised in L-L are marked by the black dots, which result in $\mathbf{a_{\mathrm{opt}}}=[1i,-1i,1,-1,-1]$ and $\mathcal{R}(\mathbf{a}_{\mathrm{opt}},\mathbf{h})=0.585$. However, the actually optimal solution is $\mathbf{a_{\mathrm{opt}}}=[1i,1i,1,-1,-1]$ with $\mathcal{R}(\mathbf{a}_{\mathrm{opt}},\mathbf{h})=0.702$. The corresponding optimal Voronoi region is the blue solid polygon (labelled as Exhaustive-II) which is generated by the points marked with circles from the set $\mathcal{S}$-II. Since none of the black dots are located within this region, hence $\mathbf{a}_{\mathrm{opt}}$ is missed by the L-L algorithm.         
\begin{center}
\begin{figure}[h]
\centerline{\includegraphics[width=0.85\linewidth]{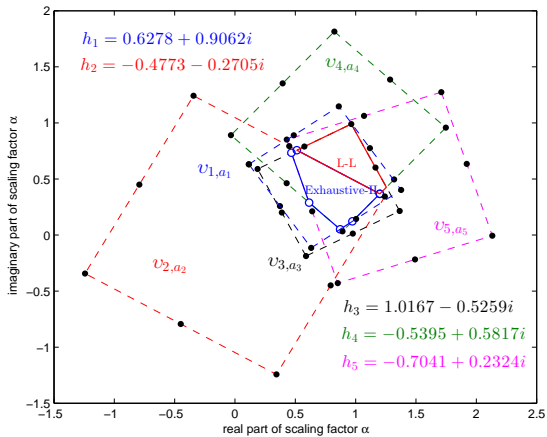}}
\caption{An example of a 5-user system, with SNR = 10dB}
\label{fig:buffer1}
\end{figure}
\end{center}

\section{Linear Search Algorithm}
In this section, we propose a simplified approach using a linear search algorithm to reduce the complexity. It maintains almost the same performance as the exhaustive search does. The complexity reduction comes from the following aspects.
\begin{itemize}
\item The exhaustive method in section III requires the calculation of the vertices of all irregular polygons in order to obtain a complete candidate set. In this section we simply employ some sampled values of $\alpha$ to acquire the candidates.
\item We set a low sampling rate (or large step size) to ignore the ``unnecessary candidates", the step size can be drawn from an off-line acquired table.
\item We set a break condition for the online search.   
\end{itemize}       

\subsection{Off-line Search: Obtain The Optimal Step Size}   
The corresponding polygons of $\mathbf{a}$ are uniformly distributed over the range of $\alpha$ with random sizes. Hence we utilise the simplest uniform sampler to generate $\alpha$ as
\begin{equation}
\alpha_{\mathrm{sample}}=\Delta(k_1+k_2i),\ k_1,k_2\in\mathbb{Z},
\end{equation}
where the positive real number $\Delta$ denotes the step size which controls the sampling rate. The key factor is to choose a proper step size.

\begin{center}
\begin{figure}[h]
\centerline{\includegraphics[width=0.95\linewidth]{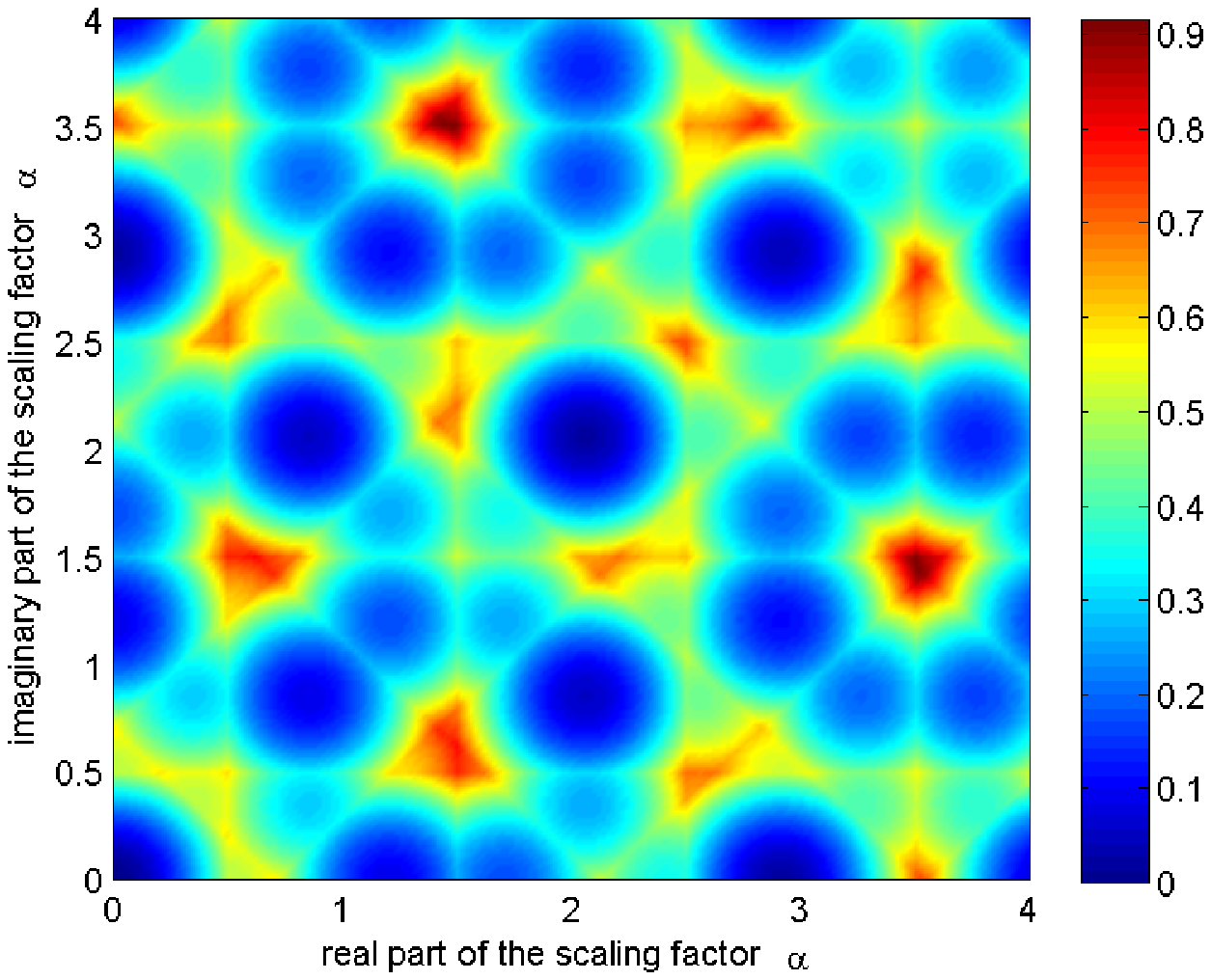}}
\caption{$\sigma^2_{\mathrm{eff}}(\alpha)$ of example.1 with $\mathrm{SNR=30dB}$}
\label{fig:buffer1}
\end{figure}
\end{center}

Fig. 4 gives an intuitive view of determining the step size. We adopt the same channel and axis labelling as in Fig. 2. Again the x-axis and y-axis denote the real and imaginary parts of $\alpha$ respectively, and the corresponding effective noise calculated by 
\begin{equation}
\sigma^2_{\mathrm{eff}}(\alpha)=\|\alpha\|^2\sigma^2+P\|\alpha\mathbf{h}-\mathcal{Q}_{\mathbb{Z}[i]}(\alpha\mathbf{h})\|^2
\end{equation} 
is shown in the colour bar. The 1st order derivative of (19) is expressed as
\begin{equation}
\frac{d\sigma^2_{\mathrm{eff}}}{d\alpha}=2\alpha\sigma^2+2P\alpha\|\mathbf{h}\|^2-2\mathbf{h}^{H}\mathcal{Q}_{\mathbb{Z}[i]}(\alpha\mathbf{h}).
\end{equation}
Since $\mathcal{Q}_{\mathbb{Z}[i]}(\alpha\mathbf{h})$ is invariant within each polygon, the 2nd order derivative is therefore expressed as 
\begin{equation}
\frac{d^2\sigma^2}{d\alpha^2}=2\sigma^2+2P\|\mathbf{h}\|^2\geq{0}.
\end{equation}
Clearly, there is a local minimum within each polygon since $\sigma^2_{\mathrm{eff}}(\alpha)$ is convex. More importantly, the 2nd derivative is the same for all candidate vectors, which means the global minimum is more likely to be located in one of the larger polygons. As shown in Fig. 4, the dark blue regions correspond to the large polygons in Fig. 2. Their corresponding $\mathbf{a}$ can be regarded as ``necessary candidates" since they have lower effective noise. 

\begin{center}
\begin{figure} [h]
\centerline{\includegraphics[width=0.8\linewidth]{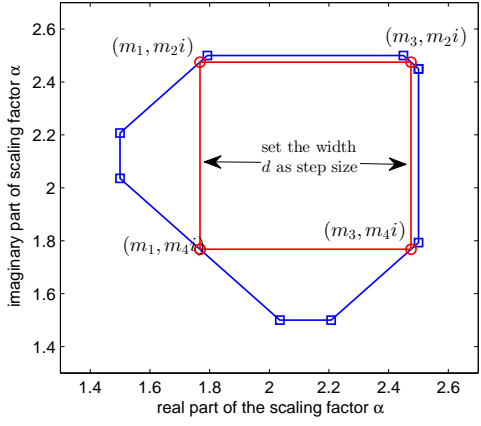}}
\caption{Finding the step size given the optimal Voronoi region}
\label{fig:buffer1}
\end{figure}
\end{center}

\par 
Let $\upsilon_{\mathrm{opt}}$ denote the corresponding Voronoi region of $\mathbf{a}_{\mathrm{opt}}$, and $\upsilon_{\mathrm{opt}}$ has $g$ edges. Actually the $\upsilon_{\mathrm{opt}}$ of the example above corresponds to the black octagon labelled in Fig. 2. Assume the largest square (with all sides vertical or horizontal) that fits in $\upsilon_{\mathrm{opt}}$ has width $d_{\mathrm{opt}}$, as shown in Fig. 5. The region $\upsilon_{\mathrm{opt}}$ will definitely be visited if $\Delta\leq{d_{\mathrm{opt}}}$. Finding the largest square in $\upsilon_{\mathrm{opt}}$ is a convex optimisation problem described as:
\begin{equation*}
\begin{aligned}
& \underset{\mathbf{m}}{\text{maximise}}
& &  \mathbf{m}^{T}\mathbf{Q}\mathbf{m} \\
& \text{subject to}
& & \mathbf{A}\mathbf{m} \leq \mathbf{b} \\
& & & \text{and} \  m_1+m_2=m_3+m_4,
\end{aligned}
\end{equation*}
where 
$\mathbf{Q}=\begin{bmatrix}
0 &1 &0 &-1\\[-0.2em]
1 &0 &-1 &0\\[-0.2em]
0 &-1 &0 &1\\[-0.2em]
-1 &0 &1 &0\\[-0.2em]
\end{bmatrix}$, and the vertices of the square are denoted by $\mathbf{m}=[m_1,m_2,m_3,m_4]^{T}$, as labelled in Fig. 5. The restriction $\mathbf{A}\mathbf{m} \leq \mathbf{b}$ comprises 4$g$ linear equations which corresponding to the condition that the 4 vertices of the square should be located within the $g$-edge convex polygon. Such an optimisation problem is linearly solvable, with complexity $\mathcal{O}(g)$.  
\begin{center}
\begin{figure} [h]
\centerline{\includegraphics[width=0.9\linewidth]{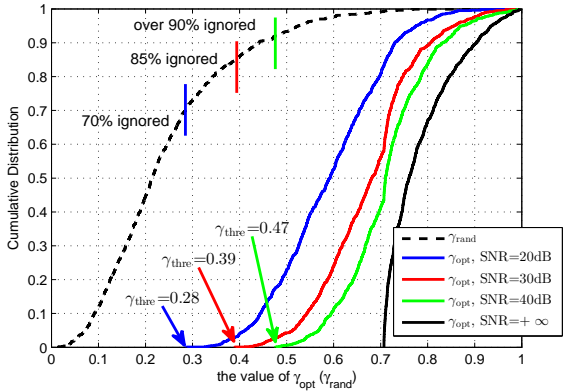}}
\caption{The cumulative distributions of $\gamma_{\mathrm{opt}}$ and $\gamma_{\mathrm{rand}}$ }
\end{figure}
\end{center}
\par
Given SNR and $L$, the optimal $\Delta$ can be determined by exploiting the statistical characteristic of $d_{\mathrm{opt}}$. As discussed in section III, the area of $\upsilon_{\mathrm{opt}}$ is upper bounded by $\mathcal{A}_{\upsilon}=\mathcal{A}_{\upsilon_{0}}/\|h_{\mathrm{max}}\|^2$ for a particular $\mathbf{h}$, hence we define the normalised $d_{\mathrm{opt}}$ as $\gamma_{\mathrm{opt}}=d_{\mathrm{opt}}/\sqrt{\mathcal{A}_{\upsilon}}$. We use $\upsilon_{\mathrm{rand}}$ to denote a random polygon within the range of $\alpha$, and $d_{\mathrm{rand}}$ denotes the width of the  largest square that fits in $\upsilon_{\mathrm{rand}}$. Similarly $\gamma_{\mathrm{rand}}=d_{\mathrm{rand}}/\sqrt{\mathcal{A}_{\upsilon}}$. It is obvious that $0<{\gamma_{\mathrm{opt}}(\gamma_{\mathrm{rand}})}\leq{1}$. Fig. 6 illustrates the cumulative distribution of $\gamma_{\mathrm{opt}}$ in a 5-user, $\mathbb{Z}[i]$-lattice based system. The results are acquired over 1000 channel realisations. The blue, red and green lines represent the scenarios of SNR=20dB, 30dB and 40dB respectively. It can be seen that $\gamma_{\mathrm{opt}}\geq{0.28}$ for all channel trails when SNR=20dB, hence we can set $\gamma_{\mathrm{thre}}=0.28$ as the threshold to distinguish the necessary and unnecessary candidates, and therefore set 
\begin{equation}
\Delta=\gamma_{\mathrm{thre}}\sqrt{\mathcal{A}_{\upsilon_{v}}}=\gamma_{\mathrm{thre}}\sqrt{\frac{\mathcal{A}_{\upsilon_{0}}}{\|h_{\mathrm{max}}\|^2}}
\end{equation} 
to capture the necessary candidates. Similarly, $\gamma_{\mathrm{thre}}=0.39$ and $\gamma_{\mathrm{thre}}=0.47$ can be assigned to SNR=30dB and 40dB respectively. We also investigated the cumulative distribution of $\gamma_{\mathrm{rand}}$, which reveals the potential complexity reduction compared to the exhaustive-II algorithm. There are over 70\% candidates whose corresponding $d\leq{0.47}$ with SNR=20dB. That means that 70\% of the candidates examined in the exhaustive search are ignored by the sampled values, hence the complexity potentially reduces by 70\% in this case\footnote{Some of these ignored candidates might still be visited by the sampled values, hence we call it potential complexity reduction.}. The corresponding $\gamma_{\mathrm{thre}}$ for SNR=30dB and 40dB indicate that the higher SNR, the more complexity reduction.  \par 

It can be observed that $\gamma_{\mathrm{thre}}$ increases monotonically with SNR, for convenience, the $\gamma_{\mathrm{thre}}$ obtained for SNR=20dB can be used in the region of 20dB$\leq\Delta<$25dB, this does not change the accuracy. Additionally, the threshold of $\gamma_{\mathrm{opt}}$ only depends on SNR and $L$, and not on any particular instance of the channel, hence an off-line table of $\Theta$ can be established to store the values of $\gamma_{\mathrm{thre}}$ corresponding to each $L$ and SNR region, which does not bring extra complexity to the online search. 
\subsection{Online Search: Obtain The Optimal Integer Vector}
Based on the table of $\Theta$, the step size for a particular channel $\mathbf{h}$ can be calculated by (22). The task of the online search is to check the candidates $\mathbf{a}=\mathcal{Q}_{\mathbb{A}}(\alpha_{\mathrm{sample}}\mathrm{h})$ one by one, and select $\alpha_{\mathrm{opt}}$. We perform the following processes to make it more efficient. 
\begin{itemize}
\item the values of $\alpha_{\mathrm{sample}}$ are sorted in ascending order of amplitude. Note that the step size $\Delta$ changes the scale in (18) only, hence the order of $\alpha_{\mathrm{sample}}$ is invariant for different $\mathbf{h}$, and no extra complexity is required.  
\item we set a break condition as follows: the search terminates when the scaled Gaussian noise ($\sigma_{\mathrm{sg}}^2=\|\alpha\|^2\sigma^2$) of the current sample is already greater than the minimum effective noise obtained from the preceding samples (It is impossible to find better $\alpha$ with larger amplitude even it brings no self noise at all). 
\end{itemize}

\begin{center}
\begin{table*}[t]
\caption{A partial table of $\Theta$ with $L=5, 8, 10$} 
\centering 
\begin{tabular}{c|c|c c c c c c c c c c|} 
\cline{2-12}
& \ 
\backslashbox{$L$}{SNR in dB} & $<$5 & [5 10) & [10 15) &[15 20)  &[20 25) &[25 30) &[30 35) &[35 40) &$\cdots$ &$+\infty$  \\ [0.3ex] 
\hline
\multicolumn{1}{ |c  }{\multirow{3}{*}{$\mathbb{Z}[i]$} } &
\multicolumn{1}{ |c| }{5} & $E$ & 0.09 & 0.12 & 0.21 &0.28 &0.33 &0.39  &0.44  &$\cdots$ &0.71 \\ \cline{2-12}
\multicolumn{1}{ |c  }{}                        &
\multicolumn{1}{ |c| }{8} & $E$ & 0.05 & 0.07 & 0.13 & 0.16 &0.25 &0.32 &0.38 &$\cdots$ &0.71 \\ \cline{2-12}
\multicolumn{1}{ |c  }{}                        &
\multicolumn{1}{ |c| }{10} & $E$ & 0.05 & 0.06 & 0.10 & 0.12 &0.17 &0.22 &0.29 &$\cdots$  &0.71    \\ \cline{1-12}
\multicolumn{1}{ |c  }{\multirow{3}{*}{$\mathbb{Z}[\omega]$} } &
\multicolumn{1}{ |c| }{5} & $E$ & 0.10 & 0.12 & 0.20 &0.29 &0.33 &0.40  &0.44  &$\cdots$ &0.71 \\ \cline{2-12}
\multicolumn{1}{ |c  }{}                        &
\multicolumn{1}{ |c| }{8} & $E$ & 0.05 & 0.08 & 0.11 & 0.16 &0.24 &0.32 &0.37 &$\cdots$ &0.71    \\ \cline{2-12}
\multicolumn{1}{ |c  }{}                        &
\multicolumn{1}{ |c| }{10} & $E$ & 0.05 & 0.06 & 0.09 & 0.13 &0.18 &0.23 &0.28 &$\cdots$  &0.71    \\ \cline{1-12}
\end{tabular}
\label{table:nonlin} 
\end{table*}
\end{center}

\begin{algorithm}
\caption{Linear search algorithm}
\begin{algorithmic}[1] 
\ENSURE optimal coefficient vector $\mathbf{a}_{\mathrm{opt}}$ \\*
$\mathbf{Offline\ Search}$: obtain table of $\Theta$ \\*
Given particular $L$ and $\mathrm{SNR}$
\FOR {$trail=1:1000$}   
\STATE generate $\mathbf{h}_{trail}\in\mathbb{C}^{L}$ 
\STATE obtain
$\mathbf{a}_{\mathrm{opt},trail}=\argmax\mathcal{R}(\mathbf{h}_{trail},\mathbf{a})$ by exhaustive-II, and acquire its corresponding $\upsilon_{\mathrm{opt},trail}$
\STATE calculate the normalised width $\gamma_{\mathrm{opt},trail}$ for $\upsilon_{\mathrm{opt},trail}$\\*
\ENDFOR 
\STATE set $\min_{trail=1}^{1000}\gamma_{\mathrm{opt},trail}\to\gamma_{\mathrm{thre}}(L,\mathrm{SNR})$\\
$\mathbf{Online\ Search}$: obtain $\mathbf{a}_{\mathrm{opt}}$ for a given $\mathbf{h}$\\*
\STATE $\Delta=\gamma_{\mathrm{thre}}\sqrt{\frac{\mathcal{A}_{\upsilon_{0}}}{\|h_{\mathrm{max}}\|^2}}$  \ (Eq.22)\ \ \  generate $\alpha_{\mathrm{sample}}$ in ascending order, denoted as $\alpha_{index}$ 
\STATE initialise $index=1,\ \sigma^2_{\mathrm{opt}}=\sigma_{\mathrm{eff}}^2(\alpha_{index})$ \ (Eq.19)\\ 
\STATE then $index=index+1,\ \sigma_{\mathrm{sg}}^2=\|\alpha_{index}\|^2\sigma^2$
\WHILE {$\sigma_{\mathrm{opt}}^2>\sigma_{\mathrm{sg}}^2$}
\IF {$\sigma_{\mathrm{eff}}^2(\alpha_{index})<\sigma_{\mathrm{opt}}^2$} 
\STATE $\alpha_{\mathrm{opt}}=\alpha_{index}$, $\sigma_{\mathrm{opt}}^2=\sigma_{\mathrm{eff}}^2(\alpha_{index})$
\ENDIF
\STATE $index=index+1$, $\sigma_{\mathrm{sg}}^2=\|\alpha_{index}\|^2\sigma^2$ 
\ENDWHILE

\STATE Return $\mathbf{a}_{\mathrm{opt}}=\mathcal{Q}_{\mathbb{A}}(\alpha_{\mathrm{opt}}\mathbf{h})$

\end{algorithmic}
\end{algorithm}
\subsection{Complexity of the Linear Search Algorithm}
The complexity of the linear search algorithm can be analysed from two perspectives. On the one hand, the proportion of candidates ignored is quite small ($\gamma_{\mathrm{thre}}\approx{0}$) in the low SNR region, and hence the complexity of the linear search can be measured by the exhaustive-II search. On the other hand, the number of candidates for the high SNR case can be expected to be
\begin{equation}
\frac{\mathrm{SNR}}{E[\Delta^2]}=\frac{\mathrm{SNR}}{\mathcal{A}_{\upsilon_{0}}}E[\frac{\|h_{\mathrm{max}}\|^2}{\gamma_{\mathrm{thre}}^2}]=\frac{\mathrm{SNR}}{\gamma_{\mathrm{thre}}^2\mathcal{A}_{\upsilon_{0}}}E[\|h_{\mathrm{max}}\|^2],
\end{equation}
where the first equality comes from (22), and the second is due to the fact that the threshold $\gamma_{\mathrm{thre}}$ tends to a constant in the high SNR region: when $\sigma^2\to{0}$, the optimal $\alpha$ is free to be chosen as the least common multiple of $\{\frac{1}{h_{l}},l=1:L\}$. In this case, the centre points of all individual $\upsilon_{l,a_{l}}$ (see Prop. 2) overlap, and hence the optimal Voronoi $\upsilon_{\mathrm{opt}}$ is very likely to be the smallest individual $\upsilon_{l}$. By employing the moment generating function of $\|h_{\mathrm{max}}\|^2$, we have
\begin{align*}
E[\|h_{\mathrm{max}}\|^2]&=\frac{1}{\beta}E[\mathrm{log}e^{\beta\|h_{\mathrm{max}}\|^2}]~~~(\beta>0)\numberthis\\&\leq\frac{1}{\beta}\mathrm{log}E[e^{\beta\|h_{\mathrm{max}}\|^2}]\numberthis\\&=\frac{1}{\beta}\mathrm{log}\int_{0}^{\infty}\mathrm{Pr}(e^{\beta\|h_{\mathrm{max}}\|^2}\geq{x})dx\numberthis\\&\leq\frac{1}{\beta}\mathrm{log}\int_{0}^{\infty}\sum_{l=1}^{L}\mathrm{Pr}(e^{\beta\|h_l\|^2}\geq{x})dx\numberthis\\&=\frac{1}{\beta}\mathrm{log}\sum_{l=1}^{L}E[e^{\beta\|h_{l}\|^2}]\numberthis\\&=\frac{1}{\beta}\mathrm{log}\frac{L}{1-2\beta}\numberthis
\end{align*}
where (25) comes from Jensen's inequality. (26) and (28) are based on the relation between the expected value and the survival function. (27) is obtained by the union bound and (29) is because the moment generating function of a chi-square variable $\|h_{l}\|^2$ is $\frac{1}{1-2\beta}$. Since (29) holds for any $\beta>0$, we can pick $\beta$ to tighten this bound. By employing the AM-GM\footnote{ Arithmetic Mean-Geometric Mean: $\frac{\sum_{i=1}^{n}a_{i}}{n}\geq(a_{1}a_{2}\cdots{a_{n}})^{1/n}$ for positive numbers $a_{i}$, the equality hold iff all the numbers are equal.} inequality (setting $\mathrm{log}L=\mathrm{log}\frac{1}{1-2\beta}$), we have $E[\|h_{\mathrm{max}}\|^2]\leq{\frac{4\mathrm{log}L}{1-\frac{1}{L}}}$. Again, the corresponding $\mathcal{R}(\mathbf{a,h})$ can be calculated in $\mathcal{O}(L)$, the time complexity for the high SNR can expressed as
\begin{equation}
\mathcal{O}(\mathrm{SNR}L\frac{\mathrm{log}{L}}{1-\frac{1}{L}})
\end{equation}  
with the constant components omitted.

\section{Numerical results}
In this section, we investigate both the computation rate and the complexity of our proposed algorithms, compared with the CLLL method \cite{CongLing.2009} and the L-L\cite{Ling.2016} algorithm. We consider two scenarios in which 5 and 10 users are employed respectively. All results are acquired over 10000 channel realisations.
\subsection{Computation Rate Comparison}
\begin{center}
\begin{figure}[h]
\centerline{\includegraphics[width=0.95\linewidth]{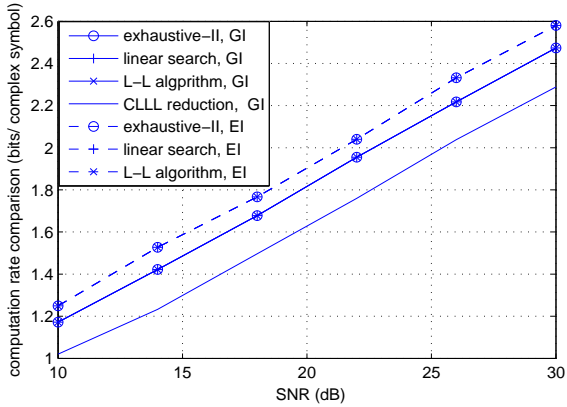}}
\caption{Average  $\mathcal{R}(\mathbf{h})$ comparison: 5 users}
\label{fig:buffer1}
\end{figure}
\end{center}
\par
Fig.7 shows the average $\mathcal{R}(\mathbf{h})$ of a 5 user scenario. We use solid and dashed lines to represent the case of $\mathbb{Z}[i]$ (denoted as GI) and $\mathbb{Z}[\omega]$ (denoted as EI) based lattices respectively. Unsurprisingly, the denser structure of $\mathbb{Z}[\omega]$ leads to a better performance than the $\mathbb{Z}[i]$ based lattice. Previously we have established that both the L-L algorithm and the linear search method might sometimes miss the optimal solution. However, the numerical results reveal that the probability of missing $\mathbf{a}_{\mathrm{opt}}$ is quite small. The gaps to the exhaustive-II algorithm are negligible for both algorithms, and they all outperform the CLLL method. Similarly, Fig. 8 reveals the rate comparison of a 10 user scenario. Compared to the case of $L=5$, the advantage of our proposed algorithms to the CLLL is increased.    
\begin{center}
\begin{figure}[h]
\centerline{\includegraphics[width=0.95\linewidth]{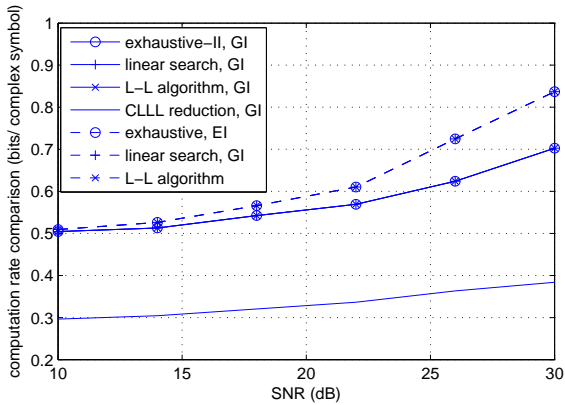}}
\caption{Average $\mathcal{R}(\mathbf{h})$ comparison: 10 users}
\label{fig:buffer1}
\end{figure}
\end{center}
\subsection{Complexity Comparison}  
In this section, we investigate the complexity by counting the floating point operations (flops). The number of flops required for each complex addition and multiplication are 2 and 6 respectively, and the round operations are ignored in the simulation. It suffices to consider $\mathbb{Z}[i]$ based lattice only (any other non-cubic lattices have a similar result). By considering $\mathbb{E}[\|\mathbf{h}\|^2]=L$, the complexity of the L-L algorithm in \cite{Ling.2016} can be rewritten as $\mathcal{O}\big(L^2(\mathrm{SNR}L+\sqrt{\mathrm{SNR}L}+2)\big)$. Compared with the expression of (17), we can see that the L-L algorithm and the exhaustive-II algorithm have almost the same theoretical complexity, both being dominated by $\mathcal{O}(L^3\mathrm{SNR})$. However, numerical results in Fig. 9 and Fig. 10 reveal that our proposed exhaustive-II algorithm has less complexity than the L-L algorithm. The reasons are as follows:
\begin{itemize}
\item the L-L algorithm considers the bound of candidate $\mathbf{a}$ as 
\begin{equation}
\|a_{l}\|\leq{\sqrt{1+\mathrm{SNR}\|\mathbf{h}\|^2}},
\end{equation} while our complex exhaustive-II considers 
\begin{equation}
\|a_{l}\|=\lfloor\alpha{h_{l}}\rceil\leq{\lfloor\sqrt{\mathrm{SNR}}h_{l}\rceil}.
\end{equation} Clearly, (32) gives a tighter bound than (31). For example, assume h = [0.3 0.4] and $\mathrm{SNR}=100$. By employing (31), we have $a_1,a_2\in[0,6]$, while (32) results in $a_1\in[0,3]$ and $a_2\in[0,4]$.
\item In section III-$C$, we have established that the $\mathcal{S}$-II set in the exhaustive-II is not considered in the L-L algorithm. However, many of the candidates $\mathbf{a}$ generated by $\lfloor\alpha\mathbf{h}\rceil,\alpha\in\mathcal{S}$-II are duplicates of the candidates generated from the set $\mathcal{S}$-I. These duplicates will not participate in the calculation of $\mathcal{R}(\mathbf{h,a})$. Hence the actual complexity of the exhaustive-II is slightly less than the expression of (17).   
\end{itemize}

\begin{center}
\begin{figure}[h]
\centerline{\includegraphics[width=0.95\linewidth]{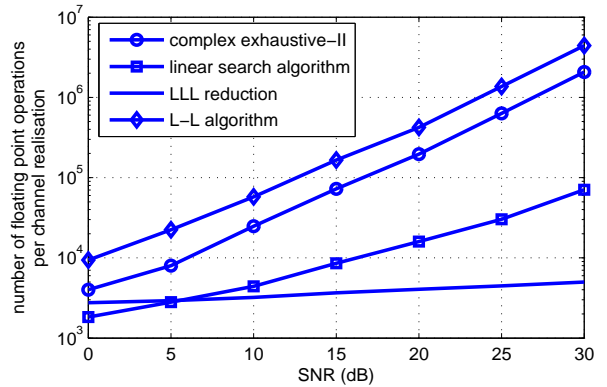}}
\caption{Average complexity comparison: 5 users}
\label{fig:buffer1}
\end{figure}
\end{center}
\par 
As we expected, the linear search has less complexity than the L-L and exhaustive-II. Since the the complexity of the linear search varies, the gap increases as the SNR increases. The comparison of the LLL and the other three is a tradeoff between $L$ and $\mathrm{SNR}$. In the high SNR region, the LLL algorithm has the complexity advantage while for a large number of users, our proposed algorithms have less complexity. 

\section{Concluding remarks}

In this paper, we have given two algorithms for coefficient selection in C\&F over complex integer based lattices. For the complex exhaustive search, we extended the idea of interval partition to Voronoi region partition to ensure the acquired coefficients are optimal. For the sub-optimal linear search algorithm, we established an off-line table to allocate the step size to eliminate unnecessary candidates. We have shown the theoretical complexity for both algorithms. Numerical comparisons with other existing algorithms are also given. We have shown both of our proposed approaches have good performance-complexity tradeoff. 
\begin{center}
\begin{figure}[h]
\centerline{\includegraphics[width=0.95\linewidth]{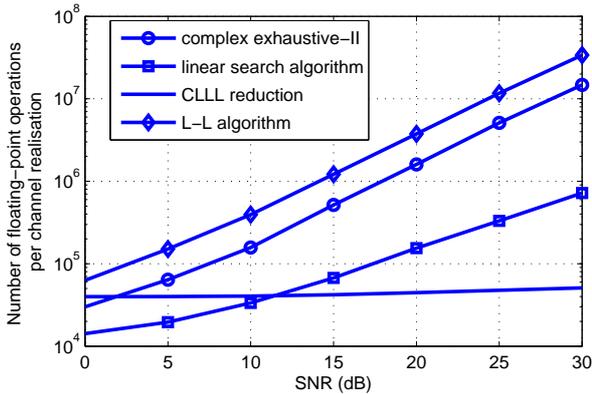}}
\caption{Average complexity comparison: 10 users}
\label{fig:buffer1}
\end{figure}
\end{center}
\bibliographystyle{IEEEtran}
\bibliography{ComplexSelection}

\begin{thebibliography}{10}
\providecommand{\url}[1]{#1}
\csname url@samestyle\endcsname
\providecommand{\newblock}{\relax}
\providecommand{\bibinfo}[2]{#2}
\providecommand{\BIBentrySTDinterwordspacing}{\spaceskip=0pt\relax}
\providecommand{\BIBentryALTinterwordstretchfactor}{4}
\providecommand{\BIBentryALTinterwordspacing}{\spaceskip=\fontdimen2\font plus
\BIBentryALTinterwordstretchfactor\fontdimen3\font minus
  \fontdimen4\font\relax}
\providecommand{\BIBforeignlanguage}[2]{{%
\expandafter\ifx\csname l@#1\endcsname\relax
\typeout{** WARNING: IEEEtran.bst: No hyphenation pattern has been}%
\typeout{** loaded for the language `#1'. Using the pattern for}%
\typeout{** the default language instead.}%
\else
\language=\csname l@#1\endcsname
\fi
#2}}
\providecommand{\BIBdecl}{\relax}
\BIBdecl

\bibitem{Zhang.2006}
\BIBentryALTinterwordspacing
S.~Zhang, S.~C. Liew, and P.~P. Lam, ``Hot topic: Physical-layer network
  coding,'' in \emph{Proceedings of the 12th Annual International Conference on
  Mobile Computing and Networking}, ser. MobiCom '06.\hskip 1em plus 0.5em
  minus 0.4em\relax New York, NY, USA: ACM, 2006, pp. 358--365. [Online].
  Available: \url{http://doi.acm.org/10.1145/1161089.1161129}
\BIBentrySTDinterwordspacing

\bibitem{Nazer.2011}
B.~Nazer and M.~Gastpar, ``Compute-and-forward: Harnessing interference through
  structured codes,'' \emph{Information Theory, IEEE Transactions on}, vol.~57,
  no.~10, pp. 6463--6486, Oct 2011.

\bibitem{Feng.2013}
C.~Feng, D.~Silva, and F.~Kschischang, ``An algebraic approach to
  physical-layer network coding,'' \emph{Information Theory, IEEE Transactions
  on}, vol.~59, no.~11, pp. 7576--7596, Nov 2013.

\bibitem{LLL.1982}
A.~K. Lenstra and H.~W. Lenstra, ``Factoring polynomials with rational
  coefficients,'' \emph{Muth. Ann}, pp. 515--534, 1982.

\bibitem{Fincke.1985}
U.~Fincke and M.~Pohst, ``Improved methods for calculating vectors of short
  length in a lattice, including a complexity analysis,'' \emph{Mathematics of
  computation}, vol.~44, no. 170, pp. 463--471, 1985.

\bibitem{Gastpar.2014}
S.~Sahraei and M.~Gastpar, ``Compute-and-forward: Finding the best equation,''
  in \emph{Communication, Control, and Computing (Allerton), 2014 52nd Annual
  Allerton Conference on}, Sept 2014, pp. 227--233.

\bibitem{Wen.2016}
J.~Wen and X.~W. Chang, ``A linearithmic time algorithm for a shortest vector
  problem in compute-and-forward design,'' in \emph{2016 IEEE International
  Symposium on Information Theory (ISIT)}, July 2016, pp. 2344--2348.

\bibitem{Qinhui.2016}
Q.~Huang and A.~Burr, ``Low complexity coefficient selection algorithms for
  compute-and-forward,'' in \emph{2016 IEEE 83rd Vehicular Technology
  Conference (VTC Spring)}, May 2016, pp. 1--5.

\bibitem{CongLing.2009}
Y.~H. Gan, C.~Ling, and W.~H. Mow, ``Complex lattice reduction algorithm for
  low-complexity full-diversity mimo detection,'' \emph{Signal Processing, IEEE
  Transactions on}, vol.~57, no.~7, pp. 2701--2710, July 2009.

\bibitem{Stewart.2001}
\BIBentryALTinterwordspacing
I.~Stewart and D.~Tall, \emph{Algebraic Number Theory and Fermat's Last
  Theorem: Third Edition}, ser. Ak Peters Series.\hskip 1em plus 0.5em minus
  0.4em\relax Taylor \& Francis, 2001. [Online]. Available:
  \url{https://books.google.co.uk/books?id=PIibasv45boC}
\BIBentrySTDinterwordspacing

\bibitem{Cohen.2013}
\BIBentryALTinterwordspacing
H.~Cohen, \emph{A Course in Computational Algebraic Number Theory}, ser.
  Graduate Texts in Mathematics.\hskip 1em plus 0.5em minus 0.4em\relax
  Springer Berlin Heidelberg, 2013. [Online]. Available:
  \url{https://books.google.co.uk/books?id=5TP6CAAAQBAJ}
\BIBentrySTDinterwordspacing

\bibitem{Tunali.2015}
\BIBentryALTinterwordspacing
N.~E. Tunali, Y.~Huang, J.~J. Boutros, and K.~R. Narayanan, ``Lattices over
  eisenstein integers for compute-and-forward,'' \emph{{IEEE} Trans.
  Information Theory}, vol.~61, no.~10, pp. 5306--5321, 2015. [Online].
  Available: \url{http://dx.doi.org/10.1109/TIT.2015.2451623}
\BIBentrySTDinterwordspacing

\bibitem{Wang.2015}
\BIBentryALTinterwordspacing
Y.~Wang, A.~G. Burr, Q.~Huang, and M.~M. Molu, ``A multilevel framework for
  lattice network coding,'' \emph{CoRR}, vol. abs/1511.03297, 2015. [Online].
  Available: \url{http://arxiv.org/abs/1511.03297}
\BIBentrySTDinterwordspacing

\bibitem{Qifu.2013}
Q.~Sun, J.~Yuan, T.~Huang, and K.~Shum, ``Lattice network codes based on
  eisenstein integers,'' \emph{Communications, IEEE Transactions on}, vol.~61,
  no.~7, pp. 2713--2725, July 2013.

\bibitem{Ling.2016}
W.~Liu and C.~Ling, ``Efficient integer coefficient search for
  compute-and-forward,'' \emph{IEEE Transactions on Wireless Communications},
  vol.~15, no.~12, pp. 8039--8050, Dec 2016.

\bibitem{Huang.2015}
\BIBentryALTinterwordspacing
Y.~Huang, K.~R. Narayanan, and P.~Wang, ``Adaptive compute-and-forward with
  lattice codes over algebraic integers,'' \emph{CoRR}, vol. abs/1501.07740,
  2015. [Online]. Available: \url{http://arxiv.org/abs/1501.07740}
\BIBentrySTDinterwordspacing

\bibitem{Liliwei.2012}
L.~Wei and W.~Chen, ``Compute-and-forward network coding design over
  multi-source multi-relay channels,'' \emph{Wireless Communications, IEEE
  Transactions on}, vol.~11, no.~9, pp. 3348--3357, September 2012.

\bibitem{Molu.2016}
\BIBentryALTinterwordspacing
M.~M. Molu, K.~Cumanan, and A.~G. Burr, ``Low-complexity compute-and-forward
  techniques for multisource multirelay networks,'' \emph{{IEEE} Communications
  Letters}, vol.~20, no.~5, pp. 926--929, 2016. [Online]. Available:
  \url{http://dx.doi.org/10.1109/LCOMM.2016.2537810}
\BIBentrySTDinterwordspacing

\bibitem{Cellfree.C}
H.~Q. Ngo, A.~Ashikhmin, H.~Yang, E.~G. Larsson, and T.~L. Marzetta,
  ``Cell-free massive mimo: Uniformly great service for everyone,'' in
  \emph{2015 IEEE 16th International Workshop on Signal Processing Advances in
  Wireless Communications (SPAWC)}, June 2015, pp. 201--205.

\bibitem{Cellfree.J}
\BIBentryALTinterwordspacing
H.~Q. Ngo, A.~E. Ashikhmin, H.~Yang, E.~G. Larsson, and T.~L. Marzetta,
  ``Cell-free massive {MIMO} versus small cells,'' \emph{CoRR}, vol.
  abs/1602.08232, 2016. [Online]. Available:
  \url{http://arxiv.org/abs/1602.08232}
\BIBentrySTDinterwordspacing

\bibitem{Qinhui.2017}
\BIBentryALTinterwordspacing
Q.~Huang and A.~G. Burr, ``Compute-and-forward in cell-free massive {MIMO:}
  great performance with low backhaul load,'' \emph{CoRR}, vol. abs/1611.06712,
  2016. [Online]. Available: \url{http://arxiv.org/abs/1611.06712}
\BIBentrySTDinterwordspacing

\bibitem{Gastpar.2015}
S.~Sahraei and M.~Gastpar, ``Polynomially solvable instances of the shortest
  and closest vector problems with applications to compute-and-forward,''
  \emph{arXiv preprint arXiv:1512.06667}, 2015.

\bibitem{Viterbo.2012}
A.~Sakzad, E.~Viterbo, Y.~Hong, and J.~Boutros, ``On the ergodic rate for
  compute-and-forward,'' in \emph{2012 International Symposium on Network
  Coding (NetCod)}, June 2012, pp. 131--136.

\end{thebibliography}
\end{document}